\documentclass[runningheads]{llncs}

\usepackage{makeidx}
\usepackage{epsfig}
\usepackage{amsmath}
\usepackage{graphicx}
\usepackage{epstopdf}
\usepackage{url,epsfig,amssymb}

\begin{document}

\pagestyle{headings}

\mainmatter

\title{A note on the clique number of complete $k$-partite graphs}

\titlerunning{A note on the clique number of complete $k$-partite graphs}

\author{Boris Brimkov}

\authorrunning{B. Brimkov}

\institute{
Computational \& Applied Mathematics, Rice University, Houston, TX 77005, USA\\
\email{boris.brimkov@rice.edu}
}

\maketitle

\begin{abstract}
In this note, we show that a complete $k$-partite graph is the only graph with clique number $k$ among all degree-equivalent simple graphs. This result gives a lower bound on the clique number, which is sharper than existing bounds on a large family of graphs.

\smallskip

{\bf Keywords:} Clique number, independence number, complete $k$-partite graph, Tur\'an graph

\end{abstract}

\section{Preliminaries}

We first recall select graph theoretic notions used in the sequel; see \cite{bondy} for further details. All graphs considered in this note are simple graphs. 

Let $G=(V,E)$ be a graph. The number of vertices and edges in $G$ are denoted by $n$ and $m$, respectively. The \emph{neighborhood} in $G$ of a vertex $v$, denoted $N(v;G)$, is the set of vertices in $G$ adjacent to $v$; the \emph{degree} of $v$ in $G$, denoted $d(v;G)$, is equal to $|N(v;G)|$. The \emph{degree sequence} of $G$, denoted $D(G)$, is the multiset of degrees of the vertices of $G$, i.e., $D(G)=\{d(v_1;G), \ldots , d(v_n;G)\}$. Two graphs $G$ and $H$ are \emph{degree equivalent}, denoted $G\simeq H$, if they have the same degree sequence. We exclude graphs with loops and multiple edges from being degree-equivalent to a given graph.

Given $S\subset V$, the \emph{induced subgraph} $G[S]$ is the subgraph of $G$ whose vertex set is $S$ and whose edge set consists of all edges of $G$ which have both ends in $S$. The \emph{complement} of $G$, denoted $\overline{G}$, is the graph on the same vertex set in which two vertices are adjacent if and only if they are not adjacent in $G$. 

The \emph{clique number} of $G$, denoted $\omega(G)$, is the cardinality of the largest clique in $G$. An \emph{independent set} in $G$ is a set of vertices no two of which are adjacent; the \emph{independence number} of $G$, denoted $\alpha(G)$, is the cardinality of the largest independent set in $G$. A \emph{complete $k$-partite graph} $K_{a_1,\ldots, a_k}$ is a graph whose vertices can be partitioned into $k$ independent sets (called \emph{parts}) with sizes $a_1,\ldots,a_k$ so that any two vertices in different parts are adjacent.

\begin{remark}
An independent set in $G$ is a clique in $\overline G$, and the complement of a complete $k$-partite graph $K_{a_1,\ldots, a_k}$ is a disjoint union of complete graphs $K_{a_1}\cup\ldots\cup K_{a_k}$. Moreover, if $G\simeq K_{a_1,\ldots, a_k}$, then $\overline{G}\simeq K_{a_1}\cup\ldots\cup K_{a_k}$. Thus, results about cliques in $k$-partite graphs can typically be restated as results about independent sets in disjoint unions of complete graphs; this duality will be employed in the next section.
\end{remark}


\section{Main results}

Complete $k$-partite graphs and their complements play a fundamental role in extremal graph theory. A notable $k$-partite graph is the \emph{Tur\'an graph} $T(n,k)$, whose parts have sizes $\lfloor n/k\rfloor$ and $\lceil n/k\rceil$; the number of edges of $T(n,k)$ is denoted $t(n,k)$. The following well-known theorem gives an upper bound on the number of edges of a $K_{k+1}$-free graph. 
\newline

\noindent \textbf{Tur\'an's Theorem \cite{aigner,turan2}.} \emph{The graph $T(n,k)=K_{\lfloor\frac{n}{k}\rfloor,\ldots,\lceil\frac{n}{k}\rceil}$ is the unique $K_{k+1}$-free graph with the maximal number $t(n,k)$ of edges.}
\newline

\noindent From Tur\'an's Theorem, it follows that among all graphs with $t(ka,k)$ edges where $a$ is some positive integer, the only $K_{k+1}$-free graph is $T(ka,k)=K_{a,\ldots,a}$. This yields the following corollary.

\begin{corollary}
Let $G\simeq K_{a,\ldots,a}$. Then, $\omega(G)=k$ if and only if $G=K_{a,\ldots,a}$.
\end{corollary}

\noindent Our main result is the following generalization of Corollary 1.

\begin{theorem}
Let $G\simeq K_{a_1,\ldots, a_k}$. Then $\omega(G)=k$ if and only if $G=K_{a_1,\ldots, a_k}$.
\end{theorem}

\noindent The proof of Theorem 1 is laid out in the next section. We will now state some related results; first, by Remark 1, Theorem 1 can be restated in terms of the independence number of disjoint cliques, as follows.

\begin{corollary}
Let $G\simeq K_{a_1}\cup\ldots\cup K_{a_k}$. Then $\alpha(G)=k$ if and only if $G=K_{a_1}\cup\ldots\cup K_{a_k}$.
\end{corollary}

\noindent The conditions of Theorem 1 and Corollary 2 are computationally easy to check, as shown in the following proposition.

\begin{proposition}
Let $G=(V,E)$ be a graph with $|V|=n$ and $|E|=m$. The following conditions can be checked with $O(m + n\log n)$ time. 
\begin{enumerate}
\item$G=K_{a_1,\ldots, a_k}$
\item$G=K_{a_1}\cup\ldots\cup K_{a_k}$
\item$G\simeq K_{a_1,\ldots, a_k}$
\item$G\simeq K_{a_1}\cup\ldots\cup K_{a_k}$
\end{enumerate}
\end{proposition}

\begin{proof}
Conditions 1 and 2 are easily verified, as it is well-known that complete $k$-partite graphs and their complements can be recognized in $O(m)$ time.

The degree sequence of $G$ can be obtained in $O(m)$ time and the cardinality of each number in the sequence can be found in $O(n\log n)$ time. Then, Condition~3 is satisfied if and only if the cardinality of each number $d$ in $D(G)$ is an integer multiple of $n-d$ and Condition 4 is satisfied if and only if the cardinality of each number $d$ in $D(G)$ is an integer multiple of $d+1$. Each of these can be checked in linear time, so the overall time complexity of verifying Conditions 3 and 4 is $O(m+n\log n)$. \qed
\end{proof}

\noindent On the other hand, the conditions of Theorem 1 and Corollary 2 are not very restrictive, in the sense that the graphs satisfying them form large families and may be quite structurally complex. For example, it is easy to see that these families of graphs have the following properties:
\begin{enumerate}
\item Arbitrary (asymptotic) density or sparsity
\item No forbidden subgraph characterization
\item No special structure like being co-graphs or perfect graphs; see Fig. 1.
\end{enumerate}

\begin{figure}
\begin{center}
\includegraphics[scale=0.45]{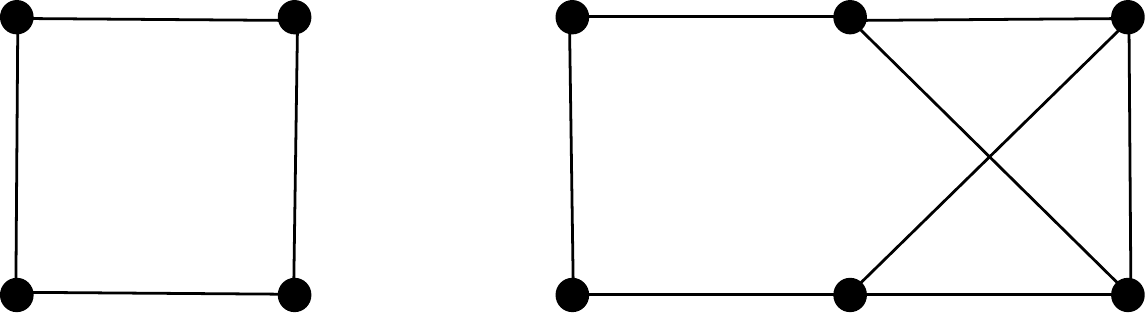}
\end{center}
\caption{A graph degree-equivalent to $K_3 \cup K_3 \cup K_4$ with independence number 4. This graph contains as induced subgraphs the path $P_4$ and the cycle $C_5$, which are forbidden induced subgraphs for co-graphs and perfect graphs.}
\end{figure}

\noindent In fact, as shown below, finding the independence and clique numbers of graphs in these families is NP-complete; thus, it is useful to have the characterizations of (k+1)-clique-free and $(k+1)$-independent set-free graphs given by Theorem~1 and Corollary~2.

\begin{proposition}
If $G\simeq K_{a_1}\cup\ldots\cup K_{a_k}$ or $G\simeq K_{a_1,\ldots,a_k}$, then finding $\alpha(G)$ and $\omega(G)$ is NP-complete. 
\end{proposition}

\begin{proof}

Let $\mathcal{P}$ be the problem of finding the independence number of a cubic graph; it is well-known that $\mathcal{P}$ is NP-complete \cite{GJ,GJ2}. Let $\mathcal{R}$ be the problem of finding the independence number of a graph which is degree equivalent to a disjoint union of cliques. We will demonstrate a polynomial reduction of $\mathcal{P}$ to $\mathcal{R}$.

Let $G=(V,E)$ be an arbitrary cubic graph with $|V|=n$. Let $G'=(V',E')$ be the disjoint union of four copies of $G$; thus $G' \simeq \bigcup_{i=1}^{n} K_4$. Obviously, the time and space needed to construct $G'$ is polynomial in $n$. Moreover, $\alpha(G)=\alpha(G')/4$, since pairwise non-adjacent vertices may be chosen independently in each copy of $G$ in $G'$. Thus, $\mathcal{R}$ is NP-complete, as well. 

Furthermore, the time and space needed to construct the complement of an $n$-vertex graph is polynomial in $n$, and the clique number of a graph is equal to the independence number of its complement. Thus, the problem of finding the clique number of a graph which is degree equivalent to a complete multipartite graph is NP-complete, as well. \qed

\end{proof}

\noindent Caro and Wei \cite{caro_wei} have shown that $\alpha(G) \geq  \sum_{i=1}^n \frac{1}{d_i+1}$, where $D(G)=\{d_1,\ldots,d_n\}$. If $G\simeq K_{a_1}\cup\ldots\cup K_{a_k}$, then $a_i$ appears $a_i+1$ times in $D(G)$, $1\leq i \leq k$; thus, the Caro-Wei bound yields $\alpha(G)\geq k$. Using this fact, Corollary 2 (and thus Theorem 1) is equivalent to the following statement.

\begin{corollary}
Let $G\simeq K_{a_1}\cup\ldots\cup K_{a_k}$. If $G\neq K_{a_1}\cup\ldots\cup K_{a_k}$, then $\alpha(G)\geq k+1$.
\end{corollary}


\noindent By Remark 1, Corollary 3 can also be restated as a bound on the clique number, as follows.


\begin{corollary}
Let $G\simeq K_{a_1,\ldots, a_k}$. If $G\neq K_{a_1,\ldots, a_k}$, then $\omega(G)\geq k+1$.
\end{corollary}

\noindent The bounds of Corollaries 3 and 4 are sharp, as shown by the graph in Fig. 1 and its complement. In contrast, it is easy to check that existing bounds like the ones below are not sharp for the families of graphs in Corollaries 3 and 4. 

\renewcommand{\arraystretch}{1.5}
\setlength{\tabcolsep}{25pt}
\begin{tabular}{ l l }
$\alpha(G) \geq  \sum_{i=1}^n \frac{1}{d_i+1}$ & Caro and Wei \cite{caro_wei} \\
$\alpha(G) \geq  \frac{n^2}{n+2m}$ & Tur\'an \cite{ajtai_erdos,griggs,turan} \\
$\alpha(G) \geq  \left\lceil \frac{2n-2m/\lfloor 2m/n \rfloor}{\lfloor 2m/n \rfloor+1}\right\rceil$ & Hansen and Zheng \cite{hansen} \\
$\omega(G)\geq\frac{n^2}{n^2-2m}$& Myers and Liu \cite{myers_liu}\\
$\omega(G)\geq n/(n-(\frac{1}{n}\sum_{i=1}^n d_i^2)^{1/2})$ & Edwards and Elphick \cite{edwards_elphick}\\
\end{tabular}
\newline

\noindent Thus, we have shown that a complete $k$-partite graph is the only graph which does not contain a $(k+1)$-clique among all degree-equivalent graphs. Equivalently, a disjoint union of $k$ cliques is the only graph which does not have a $(k+1)$-independent set among all degree-equivalent graphs. These results can be formulated as bounds on the independence and clique numbers, which are sharper than existing bounds on large families of graphs.


\section{Proof of Theorem 1}

For technical simplicity, we will prove Corollary 3, which is equivalent to Theorem 1. 

\begin{proof}
Let $G\simeq K_{a_1}\cup \ldots \cup K_{a_k}$ and $G\neq K_{a_1}\cup \ldots \cup K_{a_k}$. We want to show that $\alpha(G)\geq k+1$.

If $G$ has a connected component $Q$ which is a clique, $G-Q$ also satisfies the conditions of Corollary 3, and $\alpha(G-Q)\geq k$ if and only if $\alpha(G)\geq k+1$. Thus, without loss of generality, suppose that $G$ has no clique components.

If $a_1=\ldots = a_k$, by Corollary 1, $\alpha(G)\geq k+1$ and we are done. Thus, suppose $a_1=\ldots=a_c<a_{c+1}\leq a_{c+2}\leq \ldots \leq a_k$, where $c\geq 1$. Let $S_1,\ldots,S_k$ be a partition of the vertices of $G$, where $S_i$ has $a_i$ vertices of degree $a_i-1$. For $0\leq i \leq k-c$, let $G^{c+i}=G[S_1\cup\ldots\cup S_{c+i}]$. We will first show that $\alpha(G^c)\geq c+1$ and then by induction that $\alpha(G^{c+i})\geq c+1+i$.

Note that $G^c$ cannot have a clique component of size $a_1$, because such a component would also be a clique component in $G$, and we assumed $G$ has no clique components (there could possibly be smaller clique components in $G^c$). Also note that for any $S\subset V$ and $v\in S$, $d(v;G[S])\leq d(v;G)$; thus, $\forall v\in V(G^c)$, $d(v;G^c)\leq a_1-1$. 

Now, suppose for contradiction that $\alpha(G^c)\leq c$ and let $\mathcal{J}=\{x_1, \ldots, x_j\}$ be a maximum independent set in $G^c$, $j\leq c$. Thus, we have

\begin{equation}
\label{eq2}
\left|\bigcup_{i=1}^j N(x_i;G^c)\right|\leq \sum_{i=1}^j |N(x_i;G^c)| \leq j(a_1-1)\leq ca_1-j=|V(G^c)-\mathcal{J}|,
\end{equation}

\noindent where the first inequality is a basic fact in set theory, the second inequality follows because $d(x_i;G^c)\leq a_1-1$, and the third inequality follows because $j\leq c$.

If $|\bigcup_{i=1}^j N(x_i;G^c)|<|V(G^c)-\mathcal{J}|$, then there must be a vertex $y$ which is not adjacent to any of $x_1, \ldots, x_j$, so $\{y, x_1, \ldots, x_j\}$ is an independent set, contradicting that $\{x_1, \ldots, x_j\}$ is a maximum independent set.

If $|\bigcup_{i=1}^j N(x_i;G^c)|=|V(G^c)-\mathcal{J}|$, then all inequalities in (\ref{eq2}) must be equalities, so $j=c$ and $|\bigcup_{i=1}^j N(x_i;G^c)|= \sum_{i=1}^j |N(x_i;G^c)|$, which implies that $N(x_1;G^c), \ldots, N(x_j;G^c)$ are pairwise disjoint. Now, if $G[N(x_{\ell};G^c)]$ is not a clique for some $\ell \in \{1,\ldots, j\}$, then there are two vertices $y$ and $z$ in $N(x_{\ell};G^c)$ which are not adjacent. Then, $\{x_1, \ldots, x_{\ell-1}, y,z, x_{\ell+1}, \ldots, x_j\}$ is an independent set of size $j+1$, contradicting that $\{x_1, \ldots, x_j\}$ is a maximum independent set. Thus, $G[N(x_i;G^c)]$ must be a clique for each $1\leq i\leq j$ and hence also $G[N(x_i;G^c)\cup x_i]$ must be a clique for each $1\leq i\leq j$. But this means there are $j=c\geq 1$ clique components of size $a_1$ in $G^c$ --- a contradiction.

Thus, $\alpha(G^c)\geq c+1$, so there is an independent set $\mathcal{I}=\{x_1,\ldots,x_{c+1}\}$ in $G^c$. Recall that $a_1=\ldots=a_c$, so we can say that $d(x_1;G)\leq a_1-1$ and $d(x_{i+1};G)\leq a_i-1$ for $1\leq i \leq c$ (in fact, each of these hold with equality).

Now for the inductive step, suppose that $\mathcal{I}=\{x_1,\ldots,x_{c+j+1}\}$ is an independent set in $G^{c+j}$ for some $j\in\{0,\ldots,k-c-1\}$, and that $d(x_1;G)\leq a_1-1$ and $d(x_{i+1};G)\leq a_i-1$ for $1\leq i\leq c+j$. The vertices in $\mathcal{I}$ cannot collectively be adjacent to every vertex of $V(G^{c+j+1})-\mathcal{I}$, since 

$$\left|\bigcup_{i=1}^{c+j+1} N(x_i;G^{c+j+1})\right|\leq \sum_{i=1}^{c+j+1} |N(x_i;G^{c+j+1})|= \sum_{i=1}^{c+j+1} d(x_i;G^{c+j+1})\leq$$

$$\sum_{i=1}^{c+j+1}d(x_i;G)\leq (a_1-1)+\sum_{i=1}^{c+j}(a_i-1)<\sum_{i=1}^{c+j+1} (a_i-1)=\left|V(G^{c+j+1})-\mathcal{I}\right|.$$

\noindent The strict inequality follows from the assumption that $a_c<a_{c+1}\leq \ldots \leq a_k$.

Thus, there must be a vertex $x_{c+j+2}$ in $V(G^{c+j+1})-\mathcal{I}$ which is not adjacent to any vertex in $\mathcal{I}$. This vertex can be added to $\mathcal{I}$, so $\alpha(G^{c+j+1})\geq c+j+2$. Moreover, since $x_{c+j+2}$ is in one of $S_1,\ldots,S_{c+j+1}$, $d(x_{c+j+2};G)\leq a_{c+j+1}-1$ as required for the inductive step.

In particular, for $j=k-c-1$, this means that there is an independent set $\{x_1,\ldots,x_{k+1}\}$ in $G^k=G$, and so $\alpha(G)\geq k+1$. \qed

\end{proof}






\section*{Acknowledgements}

This research was supported by the National Science Foundation under Grant No. 1450681.

\end{document}